\documentclass{amsart}

\usepackage{amsthm}
\usepackage{amsfonts}
\usepackage{amsmath}
\usepackage{amssymb}
\usepackage{sherstov}
   \usepackage{mathptmx}
   \newcommand{\ccases}[1]{\begin{cases}#1\end{cases}}
   \newcommand{\PARENS}[1]{\left(#1\right)}
   \renewcommand{\BRACES}[1]{\left\{#1\right\}}
   
   \newcommand{\LEFTRIGHT}[3]{\left#1{#3}\right#2}
   \newcommand{\rcbrace}{\}}
   
   \newcommand{\g}{g}

\renewcommand{\mathcal}[1]{\mathscr{#1}}

\usepackage{small-caption}
\usepackage{url}

\renewcommand{\emptyset}{\varnothing}

\renewcommand{\g}{g}

\begin{document}

%\title{The Intersection of Two Halfspaces Admits \\No Nontrivial
%Sign-Representation by Polynomials}
%\title{The Trivial Sign-Representing Polynomial for the Intersection 
%of Two Halfspaces is Optimal}
\title[Sign-Representing the Intersection of Two Halfspaces by Polynomials]
{Optimal Bounds for Sign-Representing the Intersection 
\\of Two Halfspaces by Polynomials}

\author[A. A. Sherstov]{Alexander A. Sherstov$^*$}
%\address{Microsoft Research, Cambridge, MA 02142}
%\email{sherstov@cs.utexas.edu}
\date{}%\today}
\thanks{$^*$ Microsoft Research, Cambridge, MA 02142.
\; Email:
\url{sherstov@cs.utexas.edu}.}

\maketitle

\renewcommand{\baselinestretch}{1.02}
    \normalsize

\begin{abstract}
The \emph{threshold degree} of a function
$f\colon\zoon\to\moo$ is the least degree of a real
polynomial $p$ with $f(x)\equiv\sign p(x).$ We prove
that the intersection of two halfspaces on $\zoon$ has
threshold degree $\Omega(n),$ which matches the trivial
upper bound and completely answers a question due to
Klivans (2002).  The best previous lower bound was
$\Omega(\sqrt n).$ Our result shows that the
intersection of two halfspaces on $\zoon$ only admits a
trivial $2^{\Theta(n)}$-time learning algorithm based
on sign-representation by polynomials, unlike the
advances achieved in PAC learning DNF formulas and
read-once Boolean formulas.
%Previous lower bounds were: $\omega(1)$ (Minsky and Papert, 1969),
%$\Omega(\log n/\log\log n)$ (O'Donnell and Servedio, STOC'03),
%and $\Omega(\sqrt n)$ (author, FOCS'09).
The proof introduces a new technique of independent
interest, based on Fourier analysis and matrix theory.
\end{abstract}

\thispagestyle{empty}

%\renewcommand{\baselinestretch}{1.05}
 %  \normalsize

%\tableofcontents

\section{Introduction}
\label{sec:intro}

A well-studied notion in computational learning
theory is that of a \emph{perceptron}. This term stands
for the representation of a given Boolean function
$f\colon\zoon\to\moo$ in the form
$
f(x) \equiv \sign p(x)
$
for a real polynomial $p$ of some degree $d.$ The least
degree $d$ for which $f$ admits such a representation
is called the \emph{threshold degree} of $f,$ denoted
$\degthr(f).$ In other words, $\degthr(f)$ is the least
degree of a real polynomial that agrees with $f$ in
sign.  Perceptrons are appealing from a learning
standpoint because they immediately lead to efficient
learning algorithms. In more detail, let
$f\colon\zoon\to\moo$ be an unknown function of
threshold degree $d.$ Then by definition, $f$ has a
representation of the form
%$
%f(x) = \sign \big(\!
%\sum_{\abs{S}\leq d}
%%\lambda_S \prod_{i\in S} x_i\big)
%$
\begin{align*}
f(x) \equiv \sign \PARENS{
\sum_{\substack{\abs{S}\leq d}}
^{\phantom{\abs{S}\leq d}}
\lambda_S \prod_{i\in S} x_i} 
\end{align*}
for some reals $\lambda_S$ and is thus a halfspace in 
%$N=\sum_{i=0}^d {n\choose i}$ dimensions.
$N={n\choose 0}+{n\choose 1}+\cdots+{n\choose d}$
dimensions.  As a result, $f$ can be PAC learned in
time polynomial in $N,$ using any of a variety of
halfspace learning algorithms.  (Throughout this paper,
the term ``PAC learning'' refers to Valiant's standard
model~\cite{valiant84pac} of learning under arbitrary
distributions.)

The study of perceptrons dates back forty years to the
seminal monograph of Minsky and
Papert~\cite{minsky88perceptrons}, who examined the
threshold degree of several common functions.  Today,
the perceptron-based approach yields the fastest known
PAC learning algorithms for several concept classes.
One such is the class of DNF formulas of polynomial
size, posed a challenge in Valiant's original
paper~\cite{valiant84pac} and extensively studied over
the past two decades.  
%e.g.,~\cite{
%valiant85dnf,
%angluin90dnf,
%verbeurgt90dnf,
%kushilevitz91findlargecoeffs,
%blum-rudich92dnf,
%mansour92dnf,
%jackson94sieve,
%bshouty95dnf,
%bshouty96dnf,
%tt99DNF-incl-excl,
%KS01dnf}.
The fastest known algorithm for PAC learning DNF formulas
runs in time $\exp\{\tilde O(n^{1/3})\}$ and is due to
Klivans and Servedio~\cite{KS01dnf}. Specifically, the
authors of~\cite{KS01dnf} prove an upper bound of
$O(n^{1/3}\log n)$ on the threshold degree of
polynomial-size DNF formulas, which essentially matches
a classical lower bound of $\Omega(n^{1/3})$ due to
Minsky and Papert~\cite{minsky88perceptrons}.

Another success story of the perceptron-based approach
is the concept class of Boolean formulas, i.e., Boolean
circuits with fan-out $1$ at every gate. O'Donnell and
Servedio~\cite{odonnell03degree} proved an upper bound
of $\sqrt s\log^{O(d)} s$ on the threshold degree of
Boolean formulas of size $s$ and depth $d,$ giving the
first subexponential algorithm for a family of formulas
of superconstant depth.  This upper bound on the
threshold degree was improved to $s^{0.5+o(1)}$ for
any depth $d$ by Ambainis et al.~\cite{ACRSZ07nand},
building on a quantum query algorithm of Farhi et
al.~\cite{FGG08nand}. More recently,
Lee~\cite{lee09formulas} sharpened the upper bound to
$O(\sqrt s),$ which is tight.  This line of research
gives the fastest known algorithm for PAC learning
Boolean formulas.

Another extensively studied problem in computational
learning theory, and the subject of this paper, is the
problem of learning \emph{intersections of halfspaces},
i.e., conjunctions of functions of the form
$f(x)=\sign(\sum \alpha_ix_i - \theta)$ for some reals
$\alpha_1,\dots,\alpha_n,\theta.$ While solutions are
known to several restrictions of this problem~\cite{
blum-kannan97intersection-of-halfspaces,   % FOCS '93
KwekPitt:98,                               % COLT '96
Vempala:97,                                % FOCS '97
arriaga98proj,                             % FOCS '99
KOS:02,                                    % FOCS '02
KlivansServedio:04coltmargin,              % COLT '04
KLT09intersections-of-halfspaces},         % RANDOM '09
no algorithm has been discovered for PAC learning the
intersection of even two halfspaces in time faster than
$2^{\Theta(n)}.$ Progress on proving hardness results
has also been scarce. Indeed, all known hardness
results~\cite{blum92trainingNN, ABFKP:04,
focs06hardness, khot-saket08hs-and-hs} either require
polynomially many halfspaces or assume \emph{proper}
learning.  In particular, we are not aware of any
representation-independent hardness results for PAC
learning the intersection of $O(1)$ halfspaces. 

\subsection*{Our Results}

Since the perceptron-based approach yields the fastest
known algorithms for PAC learning DNF formulas and
read-once Boolean formulas, it is natural to wonder
whether it can yield any nontrivial results for the
intersection of two halfspaces. Letting $D(n)$ stand
for the maximum threshold degree over all intersections
of two halfspaces on $\zoon,$ the question becomes
whether $D(n)$ is a nontrivial (sublinear) function of
the dimension $n.$ This question has been studied by
several authors, as summarized in
Table~\ref{tab:history}. Forty years ago, Minsky and
Papert~\cite{minsky88perceptrons} used a compactness
argument to show that $D(n)=\omega(1),$ the function in
question being the intersection of two majorities on
disjoint sets variables.  O'Donnell and
Servedio~\cite{odonnell03degree} studied the same
function using a rather different approach and thereby
proved that $D(n)=\Omega(\log n/\log\log n).$ No
nontrivial upper bounds on $D(n)$ being known,
Klivans~\cite[\S7]{klivans-thesis} formally posed the
problem of proving a lower bound substantially better
than $\Omega(\log n)$ or an upper bound of $o(n).$

\begin{table}[b!]
\begin{center}
\begin{tabular}{l@{\hspace{5mm}}@{\hspace{7mm}}l}
\textit{Result} & \textit{Reference} \\
\hline
$D(n)=\omega(1)$ 
    & \cite{minsky88perceptrons}\\
$D(n)=\Omega(\log n/\log\log n)$ 
    & \cite{odonnell03degree}\\
$D(n)=\Omega(\sqrt n)$ 
    & \cite{sherstov09hshs}\\
$D(n)=\Theta(n)$ 
    & this paper
\end{tabular}
\end{center}
\caption{Lower bounds for the intersection of two
halfspaces.}
\label{tab:history}
\end{table}

It was recently shown in~\cite{sherstov09hshs} that
$D(n)=\Omega(\sqrt n),$ solving Klivans' problem and
ruling out an $n^{o(\sqrt n)}$-time PAC learning
algorithm based on perceptrons. It is clear, however,
that a PAC learning algorithm for the intersection of
two halfspaces in time $n^{\Theta(\sqrt n)}$ would
still be a breakthrough in computational learning
theory, comparable to the advances in the study of DNF
formulas and read-once Boolean formulas.  The main
contribution of this paper is to prove that
$D(n)=\Omega(n),$ which matches the trivial upper bound
and definitively rules out the perceptron-based
approach for learning the intersection of two
halfspaces in nontrivial time.

\begin{maintheorem}[Main result]
\label{thm:main}
For $n=1,2,3,\dots,$ let $D(n)$ denote the maximum threshold
degree of a function of the form $f(x)\wedge g(x),$ where
$f,g\colon\zoon\to\moo$ are halfspaces. Then
\begin{align*}
D(n) = \Theta(n).
\end{align*}
\end{maintheorem}

To be more precise, we give a randomized algorithm which 
with probability at least \mbox{$1-\e^{-n/12}$} constructs 
two halfspaces on $\zoon$ whose intersection has 
threshold degree $\Theta(n).$
In Section~\ref{sec:main}, we develop several
refinements of Theorem~\ref{thm:main}.
For example, we show that the intersection of two
halfspaces on $\zoon$ requires a perceptron with
$\exp\{\Theta(n)\}$ monomials, i.e., does not have a
sparse sign-representation. We also give an essentially
tight lower bound on the threshold degree of the
intersection of a halfspace and a majority function,
improving quadratically on the previous bound
in~\cite{sherstov09hshs}.

In summary, unlike DNF formulas and read-once Boolean
formulas,
%, whose threshold degree was settled in
%previous work~\cite{minsky88perceptrons,
%KS01dnf,
%ACRSZ07nand,
%lee09formulas},
the intersection of two halfspaces does not admit a
nontrivial sign-representation.  Apart from
computational learning theory, lower bounds on the
threshold degree have played a key role in several
works on circuit
complexity~\cite{paturi-saks94rational,
siu-roy-kailath94rational,
krause94depth2mod,         %STOC'94
KP98threshold,
sherstov07ac-majmaj},
Turing complexity
classes~\cite{aspnes91voting, % STOC '91
beigel91rational,             % STOC '91
beigel94perceptrons},         % Structures '92
and communication complexity~\cite{sherstov07ac-majmaj,
sherstov07quantum,
sherstov07symm-sign-rank,
RS07dc-dnf}.
For this reason, we consider Theorem~\ref{thm:main}
and the techniques used to obtain it to be of 
interest outside of computational learning.

Theorem~\ref{thm:main} and much previous work suggest
that the nature of a PAC learning problem changes
significantly when, instead of Valiant's original
arbitrary-distribution setting, one considers learning
with respect to restricted distributions. For example,
the uniform distribution on the sphere $\mathbb{S}^{n-1}$ 
or hypercube $\zoon$ allows the use of 
tools other than sign-representing polynomials, such as
Fourier analysis.  In particular,
polynomial-time algorithms are known for the
uniform-distribution learning of intersections of a
constant number of halfspaces on the sphere~\cite{
blum-kannan97intersection-of-halfspaces,   % FOCS '93
Vempala:97}                                % FOCS '97
and hypercube~\cite{KOS:02}. Furthermore, if
membership queries are allowed, DNF formulas are known
to be learnable in polynomial time with respect to the
uniform distribution on the
hypercube~\cite{jackson94sieve}.

\subsection*{Our Techniques}

Let $f\wedge f$ denote the conjunction of two copies of
a given Boolean function $f,$ each on an independent
set of variables.  It was shown
in~\cite{sherstov09hshs} that the threshold degree of
$f\wedge f$ equals, up to a small multiplicative
constant, the least degree of a rational function $R$
with $\|f-R\|_\infty\leq 1/3.$ With this
characterization in hand, the equality $\degthr(f\wedge
f)=\Theta(\sqrt n)$ was derived
in~\cite{sherstov09hshs} by solving the rational
approximation problem for the halfspace 
\begin{align*}
f(x) = \sign\PARENS{1 +
\sum_{i=1}^{\sqrt n} \sum_{j=1}^{\sqrt n} 2^i x_{ij}}.
\end{align*}
Unfortunately, the $\Theta(\sqrt n)$ barrier is fundamental to
the analysis in~\cite{sherstov09hshs}. To prove that in fact
$D(n)=\Theta(n),$ we pursue a rather different approach.

The intuition behind our work is as follows. Let
$\alpha_1,\alpha_2,\dots,\alpha_n$ be given nonzero
integers, and let $f\colon\zoon\to\moo$ be a given
Boolean function such that $f(x)$ is completely
determined by the sum $\sum \alpha_ix_i.$ When
approximating $f$ pointwise by polynomials and rational
functions of a given degree, can one restrict attention
to those approximants that are, like $f,$ functions of
the sum $\sum\alpha_ix_i$ alone rather than the
individual bits $x_1,x_2,\dots,x_n$?  If true, this
claim would dramatically simplify the analysis of the
threshold degree of $f$ by reducing it to a univariate
question.  Minsky and Papert~\cite{minsky88perceptrons}
showed that the claim is
indeed true in the highly special case
$\alpha_1=\alpha_2=\cdots=\alpha_n.$ For the purposes
of this paper, however, the nonzero coefficients
$\alpha_1,\alpha_2,\dots,\alpha_n$ must be of
increasing orders of magnitude and in particular must
satisfy
%in order to circumvent
%well-known \emph{positive} results due to Beigel et
%al.~\cite{beigel91rational}. In particular, we must
%have
\begin{align*}
\max_{i,j} \ABS{\frac{\alpha_i}{\alpha_j}}>\exp\{\Omega(n)\}.
\end{align*}
Minsky and Papert's argument breaks down completely in
this setting, and with good reason: coefficients
$\alpha_1,\dots,\alpha_n$ are easily
constructed~\cite{beigel94perceptrons} for which the
passage to univariate approximation increases the
degree requirement from $1$~to~$n.$

To overcome this difficulty, we use techniques from
Fourier analysis and matrix perturbation theory.
Specifically, we define an appropriate distribution on
$n$-tuples $(\alpha_1,\dots,\alpha_n)$ and study the
behavior of the sum $\sum\alpha_ix_i$ as the vector $x$
ranges over $\zoon.$ We prove that for a typical
$n$-tuple $(\alpha_1,\dots,\alpha_n$) and any
collection of sums $S\subset\ZZ$ of interest, the subset
$X_S\subset\zoon$ that induces the sums in $S$ is
highly random in that membership in $X_S$ is
uncorrelated with any polynomial of degree up to
$\Theta(n).$ With some additional work, this allows the
sought passage to a univariate question.  In
particular, we are able to prove the existence of a
halfspace $f\colon\zoon\to\moo$ such that any
multivariate rational approximant for $f$ gives a
univariate rational approximant for the sign function on
$\{\pm1,\pm2,\pm3,\dots,\pm2^{\Theta(n)}\}$ with the
same degree and error. The univariate question being
well-understood, we infer that $f$ requires a rational
function of degree $\Omega(n)$ for pointwise
approximation within $1/3$ and hence $\degthr(f\wedge
f)\geq\Omega(n)$ by the characterization
from~\cite{sherstov09hshs}.

%Multivariate real polynomials arise routinely in
%complexity theory and present a nontrivial analytic
%challenge. Indeed, several natural
%problems~\cite{nisan-szegedy94degree,
%aaronson08tutorial, lee09formulas} on multivariate
%approximation of Boolean functions remain open despite
%consistent research efforts. The technique of this
%paper, described above, is new and contributes to the
%body of methods for studying sign-representation and
%uniform approximation.

\section{Preliminaries}
\label{sec:prelim}

\subsection*{Notation.}
We will view Boolean functions as mappings $X\to\zoo$
or $X\to\moo$ for some finite set $X,$ where the output
value~$1$ corresponds to ``true'' in the former case
and ``false'' in the latter. We adopt the following
standard definition of the sign function:
\begin{align*}
  \sign x = 
  \ccases{
  -1, &x<0, \\
  0,  &x=0, \\
  1, &x>0.
  }
\end{align*}
The complement of a set $S$ is denoted $\overline S.$
We denote the symmetric difference of sets $S$ and $T$
by $S\oplus T = (S\cap \overline T) \cup (\overline S
\cap T).$ For a finite set $X,$ the symbol $\Pcal(X)$
denotes the family of all $2^{\abs{X}}$ subsets of $X.$
For functions $f,g\colon X\to\Re$ on a finite set $X,$
we use the notation
\begin{align*}
\langle f, g\rangle = \frac1{\abs X} \sum_{x\in X} f(x) g(x).
\end{align*}
We let $\log x$ stand for the logarithm of $x$ to the
base $2.$ The binary entropy function
$H\colon[0,1]\to[0,1]$ is given by $H(p) = -p\log p -
(1-p)\log (1-p)$ and is strictly increasing on
$[0,1/2].$ The following bound is well
known~\cite[p.~283]{jukna01extremal}:
\begin{align}
\sum_{i=0}^k {n\choose i} \leq 2^{H(k/n)n}, &&
k=0,1,2,\dots,\lfloor n/2\rfloor.
\label{eqn:entropy-bound}
\end{align}
For elements $x,y$ of a given set, we use the Kronecker
delta 
\begin{align*}
\delta_{x,y}
=\ccases{ 1, & x=y, \\
          0, &x\ne y.
}
\end{align*}
The symbol $P_k$ stands
for the family of all univariate real polynomials of
degree up to $k.$ The majority function
$\MAJ_n\colon\zoon\to\moo$ has the usual definition:
\begin{align*}
\MAJ_n(x)
=\ccases{-1,     & x_1+x_2+\cdots+x_n > n/2, \\
         1,     & \text{otherwise}.}
\end{align*}

\subsection*{Fourier transform.}
Consider the vector space of functions $\zoon\to\Re,$ equipped
with the inner product
\begin{align*}
\langle f,g\rangle = 2^{-n} \sum_{x\in\zoon}f(x)g(x).
\end{align*}
For $S\subseteq\oneton,$ define $\chi_S\colon \zoon\to\moo$ by
$\chi_S(x) =(-1)^{\sum_{i\in S} x_i}.$
Then $\{\chi_S\}_{S\subseteq\oneton}$ is an orthonormal
basis for the inner product space in question.  As a
result, every function $f\colon \zoon\to\Re$ has a
unique representation of the form
\begin{align*}
f=\sum_{S\subseteq\oneton} \hat f(S)\chi_S,
\end{align*} where $\hat
f(S)=\langle f,\chi_S\rangle$. The reals $\hat f(S)$ are
called the \emph{Fourier coefficients of $f.$}
The orthonormality of $\{\chi_S\}$ immediately yields
\emph{Parseval's identity}:
\begin{align}
  \sum_{S\subseteq\oneton} \hat f(S)^2 
    = \langle f,f\rangle = \Exp_{x\in\zoon}[f(x)^2]. \label{eqn:parsevals}
\end{align}
%The following fact is immediate from the definition of $\hat f(S).$
%\begin{proposition}
%Let $f\colon \zoon\to\Re$ be given. Then
%\[ \max_{S\subseteq[n]} |\hat f(S)| 
%   \leq 2^{-n} \sum_{x\in\zoon} |f(x)|.\]
% \label{prop:fourier-coeff-bound}
%\end{proposition}
%\begin{proof}
%Immediate from the definition of $\hat f(S).$
%\end{proof}

\subsection*{Matrices.}
The symbol $\Re^{m\times n}$ refers to
the family of all $m\times n$ matrices with real entries.
%The $(i,j)$th entry of a matrix $A$ is denoted by $A_{ij}.$ 
%We specify matrices by their generic entry, e.g., $A=[F(i,j)]_{i,j}.$
A matrix $A\in\Re^{n\times n}$ is called \emph{strictly
diagonally dominant} if 
\begin{align*}
\abs{A_{ii}} > 
\sum_{\substack{j=1\\j\ne i}}^n \abs{A_{ij}}, &&i=1,2,\dots,n.
\end{align*}
A well-known result in matrix perturbation theory, due to
Gershgorin~\cite{gershgorin31disks}, states that the
eigenvalues of a matrix lie in the union of certain
disks in the complex plane centered around the diagonal
entries of the matrix.  We will need the following very
special case, which corresponds to showing that the
eigenvalues are all nonzero.
\begin{theorem}[\rm Gershgorin]
\label{thm:gershgorin}
Let $A\in\Re^{n\times n}$ be strictly diagonally dominant. Then
$A$ is nonsingular.
\end{theorem}

\begin{proof}[Proof \textup{(Gershgorin)}.]
Fix a nonzero vector $x\in\Re^n$ and choose $i$ such that
$\abs{x_i}=\| x\|_\infty.$ Then by strict diagonal dominance,
\begin{align*}
\abs{(Ax)_i} = \ABS{\sum_{j=1}^n A_{ij}x_j}
\geq \abs{A_{ii}} \|x\|_\infty - 
\sum_{\substack{j=1\\j\ne i}}^n \abs{A_{ij}} \|x\|_\infty
>0,
\end{align*}
so that $Ax\ne 0.$
\end{proof}

\subsection*{Rational approximation.}

The degree of a rational function $p(x)/q(x),$ where
$p$ and $q$ are polynomials on $\Re^n,$ is the maximum
of the degrees of $p$ and $q.$ Consider a function
$f\colon X\to\moo,$ where $X\subseteq \Re^n.$ For
$d\geq0,$ define
\begin{align*}
 R(f,d) \,= \,\inf_{\rule{0pt}{7pt}p,q} \,\sup_{x\in X} 
   \left\lvert f(x) - \frac{p(x)}{q(x)} \right\rvert,
\end{align*}
where the infimum is over multivariate polynomials $p$
and $q$ of degree up to $d$ such that $q$ does not
vanish on $X.$ In words, $R(f,d)$ is the least error in
an approximation of $f$ by a multivariate rational
function of degree up to $d.$ A closely related
quantity is
\begin{align*}
 R^+(f,d) \,= \,\inf_{\rule{0pt}{7pt}p,q} \,\sup_{x\in X} 
   \left\lvert f(x) - \frac{p(x)}{q(x)} \right\rvert,
\end{align*}
where the infimum is over multivariate polynomials $p$
and $q$ of degree up to $d$ such that $q$ is positive
on $X.$ These two quantities are related in a
straightforward way:
\begin{align*}
 R^+(f,2d) \leq R(f,d) \leq R^+(f,d). 
\end{align*}
The second inequality here is trivial. The first
follows from the fact that every rational approximant
$p(x)/q(x)$ of degree $d$ gives rise to a degree-$2d$
rational approximant with the same error and a positive
denominator, namely, $\{p(x)q(x)\}/q(x)^2.$

The infimum in the definitions of $R(f,d)$ and
$R^+(f,d)$ cannot in general be replaced by a
minimum~\cite{rivlin-book}, even when $X$ is finite
subset of $\Re.$ This contrasts with the more familiar
setting of a finite-dimensional normed linear space,
where least-error approximants are guaranteed to exist.

For $S\subseteq\Re,$ we let 
\begin{align*}
 R^+(S,d) = \inf_{\rule{0pt}{7pt}p,q}\, \sup_{x\in S} 
   \left\lvert \sign x - \frac{p(x)}{q(x)} \right\rvert,
\end{align*}
where the infimum ranges over $p,q\in P_d$ such that
$q$ is positive on $S.$ The study of the rational
approximation of the sign function dates back to seminal work
by Zolotarev~\cite{zolotarev1877rational} in the late
19th century. A much later result due to
Newman~\cite{newman64rational} gives highly accurate
estimates of $R^+([-n,-1]\cup[1,n],d)$ for all $n$ and
$d.$ Newman's work in particular provides upper
bounds on $R^+(\{\pm1,\pm2,\dots,\pm n\}, d),$ which
in~\cite{sherstov09hshs} were sharpened and
complemented with matching lower bounds to the
following effect:

\begin{theorem}[\rm Sherstov]
\label{thm:R+sign}
Let $n,d$ be positive integers, 
$R=R^+(\{\pm1,\pm2,\dots,\pm n\},d).$
For $1\leq d\leq\log n,$
\begin{align*}
   \exp\BRACES{-\Theta\PARENS{\frac1{n^{1/(2d)}}}}
\leq R < \exp\BRACES{-\frac1{n^{1/d}}}. 
\end{align*}
For $\log n < d <n,$
\begin{align*}
   R=\exp\BRACES{-\Theta\PARENS{\frac{d}{\log (2n/d)}}}.
\end{align*}
For $d\geq n,$ 
\begin{align*}
   R=0.
\end{align*}
\end{theorem}

\noindent
Theorem~\ref{thm:R+sign} has the following 
corollary~\cite[Thm.~1.7]{sherstov09hshs},
in which we adopt the notation $\rdeg_\eps(f)= \min\{d :
R^+(f,d)\leq\epsilon\}.$

\begin{theorem}[\rm Sherstov]
\label{thm:approx-maj}
Let $\MAJ_n\colon \zoon\to\moo$ denote the majority function. Then
\begin{align*}
\rdeg_\epsilon(\MAJ_n)= 
\ccases{ 
\displaystyle 
\Theta\PARENS{
\log \BRACES{\frac{2n}{\log(1/\epsilon)}}
\cdot 
\log \frac1\epsilon
},
	&\qquad 2^{-n}<\epsilon<1/3,\\
\rule{0mm}{10mm}
\displaystyle 
\Theta\PARENS{1 + \frac{\log n}{\log\{1/(1-\epsilon)\}}},
	&\qquad 1/3\leq \epsilon<1.
}
\end{align*}
\end{theorem}

\subsection*{Threshold degree.} 

Let $f\colon X\to\moo$ be a given Boolean function,
where $X\subset \Re^n$ is finite.  The \emph{threshold
degree} of $f,$ denoted $\degthr(f),$ is the least
degree of a polynomial $p(x)$ such that
$f(x)\equiv\sign p(x).$ The term ``threshold degree''
appears to be due to Saks~\cite{saks93slicing}. 
Equivalent terms in the literature include ``strong
degree''~\cite{aspnes91voting}, ``voting polynomial
degree''~\cite{krause94depth2mod}, ``polynomial
threshold function degree''~\cite{odonnell03degree},
and ``sign degree''~\cite{buhrman07pp-upp}.

Given functions $f\colon X\to\moo$ and $\g\colon
Y\to\moo,$ we let the symbol $f\wedge\g$ stand for the
function $X\times Y\to\moo$ given by $(f\wedge
\g)(x,y)=f(x)\wedge \g(y).$ Note that in this notation,
$f$ and $f\wedge f$ are completely different functions,
the former having domain $X$ and the latter $X\times
X.$ An elegant observation, due to Beigel et
al.~\cite{beigel91rational}, relates the notions of
sign-representation and rational approximation for
conjunctions of Boolean functions.

\begin{theorem}[\rm Beigel, Reingold, and Spielman]
\label{thm:rational-is-possible}
Let $f\colon X\to\moo$ and $\g\colon Y\to\moo$ be given functions,
where $X,Y\subseteq\Re^n.$ Let $d$ be an integer with 
$R^+(f,d) + R^+(\g,d)<1.$
Then
\begin{align*}
\degthr(f\wedge \g) \leq 2d.
\end{align*}
\end{theorem}

\begin{proof}[Proof \textup{(Beigel, Reingold, and Spielman)}.]
Consider rational functions $p_1(x)/q_1(x)$ and $p_2(y)/q_2(y)$ of degree
at most $d$ such that $q_1$ and $q_2$ are positive on $X$ and $Y,$
respectively, and
\begin{align*}
   \sup_{X} \left| f(x) - \frac{p_1(x)}{q_1(x)}\right| + 
   \sup_{Y} \left| \g(y) - \frac{p_2(y)}{q_2(y)}\right| < 1. 
\end{align*}
Then
\begin{align*}
 f(x)\wedge \g(y) \equiv \sign\{1 + f(x)+\g(y)\}
 \equiv \sign\BRACES{ 1 + \frac{p_1(x)}{q_1(x)} +
\frac{p_2(y)}{q_2(y)} }.
\end{align*}
Multiplying the last expression by the positive quantity $q_1(x)q_2(y)$
gives $f(x)\wedge \g(y) \equiv \sign\{ q_1(x)q_2(y) + p_1(x)q_2(y)
+ p_2(y)q_1(x)\}.$
\end{proof}

We will also need a converse to
Theorem~\ref{thm:rational-is-possible}, proved
in~\mbox{\cite[Thm.~3.9]{sherstov09hshs}}.

\begin{theorem}[\rm Sherstov]
\label{thm:sherstov-degthr-R}
Let $f\colon X\to\moo$ and $\g\colon Y\to\moo$ be given functions, where
$X,Y\subset\Re^n$ are arbitrary finite sets. Assume that $f$ and
$g$ are not identically false. Let $d=\degthr(f\wedge g).$ Then 
\begin{align*}
R^+(f,4d) + R^+(g,2d) < 1.
\end{align*}
\end{theorem}

\subsection*{Symmetric functions.}
Let $S_n$ denote the symmetric group on $n$ elements.
For $\sigma\in S_n$ and $x\in\zoon$, we denote $\sigma
x=(x_{\sigma(1)},\ldots,x_{\sigma(n)})\in\zoon.$ For
$x\in\zoon,$ we define $|x|= x_1+x_2+\cdots+x_n.$ A
function $\phi\colon \zoon\to\Re$ is called
\emph{symmetric} if $\phi(x) = \phi(\sigma x)$ for
every $x\in\zoon$ and every $\sigma\in S_n.$
Equivalently, $\phi$ is symmetric if $\phi(x)$ is
uniquely determined by $|x|.$ Symmetric functions on
$\zoon$ are intimately related to univariate
polynomials, as borne out by Minsky and Papert's
\emph{symmetrization
argument}~\cite{minsky88perceptrons}:

\begin{proposition}[\rm Minsky and Papert]
Let $\phi\colon \zoon\to\Re$ be a polynomial of degree $d.$ Then
there is a polynomial $p\in P_d$ such that 
\begin{align*}
  \Exp_{\sigma\in S_n} [\phi(\sigma x)] = p(|x|),
	&&x\in\zoon.
\end{align*}
\end{proposition}

\noindent
We will need the following consequence of Minsky
and Papert's technique for rational functions, 
pointed out in~\cite[Prop.~2.7]{sherstov09hshs}.

\begin{proposition}
\label{prop:symm-rational}
Let $n_1,\dots,n_k$ be positive integers.  Consider a function
$F\colon\zoo^{n_1}\times\cdots\times\zoo^{n_k}\to\moo$
such that $F(x_1,\dots,x_k)\equiv 
f(\abs{x_1},\dots,\abs{x_k})$ 
for some
$f\colon\{0,1,\dots,n_1\}\times\cdots\times\{0,1,\dots,n_k\}\to\moo.$
Then for all $d,$
\begin{align*}
R^+(F,d) = R^+(f,d).
\end{align*}
\end{proposition}

\section{Analysis of Random Halfspaces}
\label{sec:analysis-of-random-hs}

In this section, we prove a certain structural property
of random halfspaces.
Specifically, we will 
fix integers $w_1,w_2,\dots,w_n$ at random from a suitable
range and analyze the sum
\begin{align*}
\sum_{i=1}^n w_ix_i
\end{align*}
as $x$ ranges over $\zoon.$ Our objective will be to
show that, for a typical choice of the weights $w_1,w_2,\dots,w_n,$ the distribution of this sum modulo
$2^{\Theta(n)}$ is highly random. More precisely, we
will show that the subset $X_s\subset\zoon$ that
induces any particular sum $s$ modulo $2^{\Theta(n)}$
is relatively large and that membership in $X_s$ is
almost uncorrelated with any polynomial of low degree.
We start with a technical lemma.

\begin{lemma}
\label{lem:random-xor}
Let $f,g\colon\zoon\to\zoo$ be given functions. Fix an integer
$k$ with $0\leq k\leq n/2.$ For a set $S\subseteq\oneton,$ define
$F_S\colon\zoon\to\zoo$ by
\begin{align*}
F_S(x) = f(x) \wedge \left(g(x)\oplus \bigoplus_{i\in S}
x_i\right).
\end{align*}
Fix a real $\zeta >0.$ Then with probability at least
$1-2^{-n+H(k/n)n +2\zeta n}$ over a uniformly random choice
of $S\in\Pcal(\oneton),$ one has
\begin{align}
\ABS{\hat F_S(T) - \frac12 \hat f(T)}  
  \leq 2^{-\zeta n-1}, && |T|\leq k.
\end{align}
\end{lemma}

\begin{proof}
Define $\phi\colon\zoon\to[-1/2,1/2]$ by
$\phi(x)=f(x)g(x) - \frac12 f(x).$
Define $\Scal\subseteq\Pcal(\oneton)$ by 
$\Scal = \{ S : \abs{\hat \phi(S)} \geq 2^{-\zeta n-1}\}.$
By Parseval's identity~(\ref{eqn:parsevals}), 
\begin{align}
\abs{\Scal} \leq 4^{\zeta n}.
\label{eqn:A-B-bound}
\end{align}
Since 
$F_S(x) = \frac12 f(x) + (-1)^{\sum_{i\in S}x_i}\phi(x),$
we have 
\begin{align}
\ABS{\hat F_S(T) - \frac12 \hat f(T)}
 =\abs{\hat \phi(S\oplus T)},
&&S,T\subseteq\oneton.
\label{eqn:F-S-spectrum}
\end{align}

For a uniformly random $S\in\Pcal(\oneton),$ the set $\{S\oplus
T: \abs{T}\leq k\}$ contains any fixed element of
$\Pcal(\oneton)$ with probability
$2^{-n}\sum_{i=0}^k{n\choose i}.$ By the union bound, we infer
that
\begin{align*}
\Prob_S[\{S\oplus T:\abs{T}\leq k\}\cap \Scal \ne
\varnothing]
\leq \abs{\Scal}\,2^{-n}\sum_{i=0}^k {n\choose i},
\end{align*}
which in view of (\ref{eqn:entropy-bound}) and
(\ref{eqn:A-B-bound}) is bounded from above by
$2^{-n+H(k/n)n+2\zeta n}.$ This observation, along with
(\ref{eqn:F-S-spectrum}), completes the proof.
\end{proof}

Using Lemma~\ref{lem:random-xor} and induction, we now obtain
a key intermediate result.

\begin{lemma}
\label{lem:resheto}
Fix an integer $k\geq0$ and reals $\epsilon,\zeta\in(0,1/2).$
Choose sets $S_0,S_1,\dots,S_k\in\Pcal(\oneton)$ uniformly at
random.  Fix any integer $s$ and define $f\colon\zoon\to\zoo$ by 
\begin{align}
f(x)=1\qquad\Leftrightarrow\qquad \sum_{i=0}^k 2^i\sum_{j\in
S_i} x_j\equiv s \pmod{2^{k+1}}.
\label{eqn:f-modular}
\end{align}
Then with probability at least
$1-(k+1)2^{-n+H(\epsilon)n+2\zeta n}$ over
the choice of $S_0,S_1,\dots,S_k,$ one has
\begin{align}
&\ABS{\hat f(T) - \frac{\delta_{T,\emptyset}}{2^{k+1}}} \leq
2^{-\zeta n}, 
		&&\abs{T}\leq \epsilon n.
		\label{eqn:f-low-order}
\end{align}
\end{lemma}

\begin{proof}
In view of the modular counting in (\ref{eqn:f-modular}), one may
assume that $0\leq s<2^{k+1}$ and therefore 
$s=\sum_{i=0}^k 2^ib_i$ for some $b_0,b_1,\dots,b_k\in\zoo.$ The
proof of the lemma is by induction on $k$ for a fixed $s.$

The base case $k=0$ corresponds to
$f(x) = \frac12 + \frac12(-1)^{b_0}\chi_{S_0}(x).$
%	\hat f(T) = 
%		\ccases{1/2,          & T=\emptyset,\\
%			(-1)^{b_0}/2,   & T=S_0, \\
%			0,                & \text{otherwise.}
%}
%\end{align*} 
One obtains (\ref{eqn:f-low-order}) by conditioning
on the event $|S_0|>\epsilon n,$ which in view of
(\ref{eqn:entropy-bound}) occurs with probability no
smaller than 
$1-2^{-n + H(\epsilon)n}.$

We now consider the inductive step. 
Define $f'\colon\zoon\to\zoo$ by 
\begin{align*}
f'(x)=1
  \qquad\Leftrightarrow\qquad 
\sum_{i=0}^{k-1} 2^i\sum_{j\in S_i} x_j\equiv 
\sum_{i=0}^{k-1} 2^ib_i\pmod{2^{k}}.
\end{align*}
Let $E_1$ be the event, over the choice of
$S_0,\dots,S_{k-1},$ that $\abs{\widehat
{f'}(T)-2^{-k}\delta_{T,\emptyset}} \leq
2^{-\zeta n}$ for
$\abs{T}\leq \epsilon n.$ By the inductive hypothesis,
\begin{align}
\Prob[E_1] \geq 1-k2^{-n+H(\epsilon)n + 2\zeta n}.
\label{eqn:inductive-assumption}
\end{align}
Let $E_2$ be the event, over the choice of $S_0,\dots,S_k,$ that
$\abs{\hat f(T) - \frac12 \widehat{f'}(T)}
\leq 2^{-\zeta n-1}$ for $\abs{T}\leq \epsilon n.$ 
In this terminology, it suffices to show that
\begin{align}
\Prob[E_1\wedge E_2] \geq 1- (k+1)2^{-n+H(\epsilon)n +2\zeta n}.
\label{eqn:claim-rephrased}
\end{align}
Observe that
\begin{align*}
f(x) = f'(x)\wedge 
\PARENS{g(x)\oplus \bigoplus_{i\in S_k}^{\phantom{a}} x_i},
\end{align*}
where $g\colon\zoon\to\zoo$
is the function such that $g(x)=1$ if and only if $b_k$
is the $(k+1)$st least significant bit of the integer
$\sum_{i=0}^{k-1} 2^i
\sum_{j\in S_i} x_j.$ 
%In other words, $g(x)=1$ if and
%only if 
%\begin{align*}
%\sum_{i=0}^{k-1} 2^i
%\sum_{j\in S_i} x_j = 2^{k+1}A +2^kb_k + B
%\end{align*}
%for some nonnegative integers $A,B$ with $B<2^k,$
%where the integers $A$ and $B$ may depend on $x$.
As a result, Lemma~\ref{lem:random-xor} shows that
$\Prob[E_2]\geq 1-2^{-n+H(\epsilon)n +2\zeta n}.$ This bound,
along with (\ref{eqn:inductive-assumption}), settles
(\ref{eqn:claim-rephrased}) and thereby completes the induction.
\end{proof}

We have reached the main result of this section.

\begin{theorem}[\rm Key property of random halfspaces]
\label{thm:RESHETO}
Fix an integer $k\geq0$ and reals $\epsilon,\zeta\in(0,1/2).$
Choose integers $w_1,w_2,\dots,w_n$ uniformly at random from
$\{0,1,\dots,2^{k+1}-1\}.$ 
For $s\in\ZZ,$ define $f_s\colon\zoon\to\zoo$ by 
\begin{align}
f_s(x)=1\qquad\Leftrightarrow\qquad \sum_{i=1}^n w_ix_i\equiv 
s\pmod{2^{k+1}}.
\label{eqn:f_s}
\end{align}
Then with probability at least
$1-(k+1)2^{-n+H(\epsilon)n+2\zeta n+k+1}$ over
the choice of $w_1,w_2,\dots,w_n,$ one has 
\begin{align*}
&\ABS{\hat f_s(T) - \frac{\delta_{T,\emptyset}}{2^{k+1}}} \leq
2^{-\zeta n}, 
		&&\abs{T}\leq \epsilon n, \quad s\in\ZZ.
\end{align*}
\end{theorem}

\begin{proof}
In view of the modular counting in (\ref{eqn:f_s}), it
suffices to prove the theorem for
$s\in\{0,1,\dots,2^{k+1}-1\}.$ The functions $f_s$ have
the following equivalent definition: pick sets
$S_0,S_1,\dots,S_k\in\Pcal(\oneton)$ uniformly at
random and define 
\begin{align*}
f_s(x)=1\qquad\Leftrightarrow\qquad \sum_{i=0}^k2^i\sum_{j\in
S_i} x_j \equiv s\pmod{2^{k+1}}.
\end{align*}
The proof is now complete by
Lemma~\ref{lem:resheto} and the union bound over $s.$
\end{proof}

\section{Zeroing out Correlations by a Change of Distribution}
\label{sec:change-of-distribution}

Recall the setting of the previous section, where we
fixed integers $w_1,w_2,\dots,w_n$ at random from a
suitable range and analyzed the sum
$
\sum_{i=1}^n w_ix_i
$
as $x$ ranged over $\zoon.$ We showed that the subset
$X_s\subset\zoon$ that induces any particular sum $s$
modulo $2^{\Theta(n)}$ is relatively large and that
membership in $X_s$ has \emph{almost} zero correlation
with any given polynomial of low degree.  For the
purposes of this paper, the correlations with
low-degree polynomials need to be \emph{exactly} zero.
In this section we show that, with respect to a
suitable distribution $\mu_s$ on each $X_s,$ membership
in $X_s$ will indeed have zero correlation with any
low-degree polynomial.

A starting point in our discussion is a general
statement on zeroing out the correlations of given
Boolean functions $\chi_1,\chi_2,\dots,\chi_k$ with
another Boolean function $f.$ Recall that for functions
$f,g\colon X\to\Re$ on a finite set $X,$ we use the
notation
\begin{align*}
\langle f, g\rangle = \frac1{\abs X} \sum_{x\in X} f(x) g(x).
\end{align*}

\begin{theorem}
\label{thm:distribution-by-inversion}
Let $f,\chi_1,\dots,\chi_k\colon X\to\moo$ be given
functions on a finite set $X.$ Suppose that
\begin{align}
&\sum_{\substack{i=1}}^{k} \abs{\langle f,\chi_i\rangle}
<\frac12,
\label{eqn:f-chi-bounded}\\
&\sum_{\substack{j=1\\j\ne i}}^{k} \abs{\langle \chi_i,\chi_j\rangle} \leq
\frac12, &&i=1,2,\dots,k.
\label{eqn:diag-dominance}\\
\intertext{Then there exists a probability distribution $\mu$ on $X$ such
that}
&\Exp_\mu\,[f(x)\chi_i(x)] = 0, &&i=1,2,\dots,k.
\nonumber
\end{align}
\end{theorem}

\begin{remark}
A comment is in order on the hypothesis of
Theorem~\ref{thm:distribution-by-inversion}. The
theorem states that if $\chi_1,\chi_2,\dots,\chi_k$
each have a small correlation with $f$ and, in
addition, have small pairwise correlations, then a
distribution exists with respect to which $f$ is
completely uncorrelated with
$\chi_1,\chi_2,\dots,\chi_k.$ The latter part of the
hypothesis, namely the requirement
(\ref{eqn:diag-dominance}) of small pairwise
correlations for $\chi_1,\chi_2,\dots,\chi_k,$ may seem
unnecessary at first.  In actuality, it is vital.
Exponential lower bounds on the weights of linear
perceptrons~\cite{myhill-kautz61, siu91small-weights}
imply, by linear programming duality, the existence of
functions $f,\chi_1,\chi_2,\dots,\chi_k\colon X\to\moo$
such that $\abs{\langle f,\chi_i\rangle} =
\exp\{-\Theta(k)\},$ $i=1,2,\dots,k,$ and yet
\begin{align}
f(x) \equiv \sign \PARENS{\sum_{i=1}^k \alpha_i\chi_i(x)}
\label{eqn:linear-comb}
\end{align}
for some fixed reals $\alpha_1,\dots,\alpha_k.$ In this
construction, the correlation of $f$ with each $\chi_i$
is small, in fact exponentially smaller than
what is assumed in
Theorem~\ref{thm:distribution-by-inversion};
nevertheless, the representation
(\ref{eqn:linear-comb}) rules out a distribution $\mu$
with respect to which $f$ could have zero correlation
with each $\chi_i,$ for such a distribution $\mu$ would
have to obey
\begin{align*}
0<\Exp_\mu\left[\left|\sum_{i=1}^k
\alpha_i\chi_i(x)\right|\right]=\Exp_\mu\left[f(x)\sum_{i=1}^k
\alpha_i\chi_i(x)\right]=\sum_{i=1}^k\alpha_i\Exp_\mu[f(x)\chi_i(x)]=0.
\end{align*}
\end{remark}

\begin{proof}[Proof of
Theorem~\textup{\ref{thm:distribution-by-inversion}}.]
Consider the linear system
\begin{align}
M\alpha = \gamma
\label{eqn:lin-system}
\end{align}
in the unknown $\alpha\in\Re^k,$ where 
$M=[\langle \chi_i,\chi_j\rangle]_{i,j}$ is a matrix of order $k$ and
$\gamma = (\langle f,\chi_1\rangle, \dots, \langle f,\chi_k\rangle)\in\Re^k.$
Then (\ref{eqn:diag-dominance}) shows that $M$ is strictly diagonally
dominant and hence nonsingular by Theorem~\ref{thm:gershgorin}.
Fix the unique solution $\alpha$ to the system (\ref{eqn:lin-system}).
Then $2\abs{\alpha_i} - \sum_{j=1}^{k} \abs{\alpha_j\langle
\chi_i,\chi_j\rangle} \leq \abs{\langle f,\chi_i\rangle}$ for
$i=1,2,\dots,k.$ Summing
these $k$ inequalities, we obtain
\begin{align*}
2\sum_{i=1}^k \abs{\alpha_i} - \sum_{j=1}^{k}\abs{\alpha_j}
\sum_{i=1}^k\abs{\langle \chi_i,\chi_j\rangle} \leq
\sum_{i=1}^k\abs{\langle f, \chi_i\rangle}, 
\end{align*}
which in view of (\ref{eqn:f-chi-bounded}) and
(\ref{eqn:diag-dominance}) shows that $\sum_{i=1}^k\abs{\alpha_i}<1.$
Therefore, the function $\mu\colon X\to\Re$ given by
\begin{align*}
\mu(x) = \epsilon\PARENS{
 1 - f(x)\sum_{i=1}^k \alpha_i \chi_i(x)
 }
\end{align*}
is a probability distribution on $X$ for a suitable normalizing
factor $\epsilon>0.$ At last, 
\begin{align*}
\Exp_{\mu}\,[f(x)\chi_i(x)] =
\epsilon \abs{X} \left(
\langle f,\chi_i\rangle - \sum_{j=1}^k \alpha_j\langle
\chi_i,\chi_j\rangle \right) = 0,
\end{align*}
where the final equality holds by (\ref{eqn:lin-system}).
\end{proof}

We are now in a position to prove the main result of this
section.

\begin{theorem}%[\rm Distribution family]
\label{thm:mu_s-constructed}
Let $\alpha>0$ be a sufficiently small absolute constant.
Choose integers $w_1,w_2,\dots,w_n$ uniformly at random from
$\{0,1,\dots,2^{\lfloor\alpha n\rfloor+1}-1\}.$ For $s\in\ZZ,$ 
define
\begin{align}
\label{eqn:X_s}
X_s = \left\{
 x\in\zoon :  \sum_{i=1}^n w_ix_i \equiv s \pmod{2^{\lfloor\alpha
 n\rfloor+1}}\;
     \right\}. 
\end{align}
Then with probability at least $1-\e^{-n/3}$ over the choice of
$w_1,w_2,\dots,w_n,$ there is a distribution $\mu_s$
on $X_s$ $($for each $s)$ such that
\begin{align}
\label{eqn:mu_s-equivalence}
\Exp_{\mu_s}\, [p(x)] = \Exp_{\mu_t}\, [p(x)]
\end{align}
for any $s,t\in\ZZ$ and any polynomial $p$ of degree at most 
$\lfloor \alpha n\rfloor.$
\end{theorem}

\begin{proof}
Let $\alpha>0$ be sufficiently small.  We will assume throughout the proof 
that $n\geq 1/\alpha,$ the theorem being trivial
otherwise.  Set $\epsilon=2\alpha,$\; $\zeta=1/5,$\;
and $k=\lfloor \alpha n\rfloor$ in
Theorem~\ref{thm:RESHETO}. Then with probability at least
$1-\e^{-n/3}$ over the choice of $w_1,w_2,\dots,w_n,$ one has
\begin{align}
&\ABS{\hat f_s(T) - \frac{\delta_{T,\emptyset}}{2^{\lfloor \alpha
n\rfloor+1}}} \leq
2^{-n/5}, 
		&&\abs{T}\leq 2\alpha n,\quad s\in\ZZ,
		\label{eqn:f_s-guarantee}
\end{align}
where $f_s\colon\zoon\to\zoo$ is given by
$f_s(x)=1\Leftrightarrow x\in X_s.$ It follows that for each $s,$
\begin{align}
\label{eqn:X_s-size}
\abs{X_s} = 2^n\hat f_s(\emptyset) \geq
2^n(2^{-\lfloor\alpha n\rfloor-1}-2^{-n/5}).
\end{align}

For $f,g\colon\zoon\to\Re,$ we will write
$\langle f,g\rangle_{X_s} = \abs{X_s}^{-1}\sum_{x\in X_s}f(x)g(x).$
Let $\Scal\subset\Pcal(\oneton)$ be the system of nonempty
subsets of at most $\alpha n$ elements. Fix any $T\in \Scal.$
Then for each $s,$
\begin{align}
\sum_{\substack{S\in\Scal\\S\ne T}} \abs{\langle \chi_S,\chi_T\rangle_{X_s}}
=\frac{2^{n}}{\abs{X_s}}
   \sum_{\substack{S\in\Scal\\S\ne T}} \abs{\hat f_s(S\oplus T)}
\leq \frac{2^{n}}{\abs{X_s}} \cdot\abs{\Scal}\, 2^{-n/5} <\frac12,
\label{eqn:cross-correlations}
\end{align}
where the final two inequalities follow from
(\ref{eqn:entropy-bound}), (\ref{eqn:f_s-guarantee}), and (\ref{eqn:X_s-size}).
Similarly, for each $s,$
\begin{align}
\sum_{S\in\Scal} \abs{\langle f_s,\chi_S\rangle_{X_s}}
=\frac{2^{n}}{\abs{X_s}}
   \sum_{S\in\Scal} \abs{\hat f_s(S)}
\leq \frac{2^{n}}{\abs{X_s}} \cdot\abs{\Scal}\, 2^{-n/5} <\frac12.
\label{eqn:correlations-with-f_s}
\end{align}
In view of (\ref{eqn:cross-correlations}) and 
(\ref{eqn:correlations-with-f_s}), 
Theorem~\ref{thm:distribution-by-inversion} provides a
distribution $\mu_s$ on $\zoon$ that is supported on $X_s$
and obeys $\hat\mu_s(S) = 0$ for $S\in\Scal.$
Since $\mu_s$ is a probability distribution, we additionally
have $\hat\mu_s(\emptyset)=2^{-n}$ for all $s.$ In particular, the 
distributions $\mu_s$ have identical Fourier
spectra up to coefficients of order $\alpha n,$
which is another way of stating (\ref{eqn:mu_s-equivalence}).
\end{proof}

\section{Reduction to a Univariate Problem}
\label{sec:univariate-reduction}

Recall from the Introduction that the crux of our proof
is to establish the existence of a halfspace
$f\colon\zoon\to\moo$ that requires a rational function
of degree $\Theta(n)$ for pointwise approximation
within $1/3.$ The purpose of this section is to reduce
this task, for a suitably chosen random halfspace, to a
univariate problem. The univariate problem pertains to
the uniform approximation of the sign function on the
set $\{\pm1,\pm2,\pm3,\dots,\pm2^{\Theta(n)}\}$ and has
been solved in previous work. Key to this univariate
reduction will be the construction of probability
distributions in the previous two sections.

\newcommand{\esses}{\pm1,\pm2,\pm 3,\dots,\pm2^k}

\begin{theorem}[\rm Reduction to a univariate problem]
\label{thm:univariate-reduction-symm}
Put $k=\lfloor\alpha n\rfloor,$ where $\alpha>0$ is the absolute
constant from Theorem~\textup{\ref{thm:mu_s-constructed}}.
Choose $w_1,w_2,\dots,w_n$ uniformly at random from 
$\{0,1,\dots,2^{k+1}-1\}.$ Define
$f\colon\zoon\times\{0,1,2,\dots,n\}\to\moo$ by 
\begin{align*}
f(x,t) = \sign \left( \frac12 + \sum_{i=1}^n w_ix_i - 2^{k+1}t\right).
\end{align*}
Then with probability at least $1-\e^{-n/3}$ over the choice of
$w_1,w_2,\dots,w_n,$ one has 
\begin{align}
R^+(f,d)\geq R^+(\{\esses\},d), && d=0,1,\dots,k.
\end{align}
\end{theorem}

\begin{proof}
For $s=\esses,$ define $X_s\subseteq\zoon$ by (\ref{eqn:X_s}).
Then by Theorem~\ref{thm:mu_s-constructed}, with probability at
least $1-\e^{-n/3}$ there is a distribution $\mu_s$ on $X_s$ for
each $s$ such that 
\begin{align}
\label{eqn:mu_s-equivalence-restated}
\Exp_{\mu_s}\, [p(x)] = \Exp_{\mu_r}\, [p(x)]
\end{align}
for any $s,r\in\{\esses\}$ and any polynomial $p$ of
degree no greater than $k.$ In the remainder of the
proof, we will work with a fixed choice of weights
$w_1,w_2,\dots,w_n$ for which the described
distributions $\mu_s$ exist.

Suppose that
$R^+(f,d)<\epsilon$ where $0<\epsilon<1$ and $0\leq d\leq k.$
Then there are degree-$d$ polynomials $p,q$ on
$\Re^n\times \Re$ such that on the domain of $f,$
\begin{align}
0<(1-\epsilon)q(x,t)\leq p(x,t)f(x,t)\leq (1+\epsilon)q(x,t).
\label{eqn:multivariate-approximant}
\end{align}
On the support of $\mu_s$ (for $s=\esses$), the linear form
\begin{align*}
\ell(x,s) = 2^{-k-1}\PARENS{\sum_{i=1}^nw_ix_i-s}
\end{align*}
obeys $\ell(x,s)\in\{0,1,2,\dots,n\}$ and
$f(x,\ell(x,s))=\sign s.$ 
Letting $t=\ell(x,s)$ in 
(\ref{eqn:multivariate-approximant}) and passing to
expectations,
\begin{align*}
0<\Exp_{x\sim\mu_s}
\left[q(x,\ell(x,s))\right]
(1-\epsilon)
\leq &\Exp_{x\sim\mu_s}
\left[p(x,\ell(x,s))\right]\sign s\\
\leq &\Exp_{x\sim\mu_s}
\left[q(x,\ell(x,s))\right]
(1+\epsilon).
\end{align*}
It follows from (\ref{eqn:mu_s-equivalence-restated}) that
$
\Exp_{\mu_s} [p(x,\ell(x,s))] = P(s)$ and 
$\Exp_{\mu_s}[q(x,\ell(x,s))] = Q(s)$
for some $P,Q\in P_d$ and all $s.$ As a result,
$R^+(\{\esses\},d)\leq\epsilon,$ the approximant in question being
$P/Q.$
\end{proof}

It remains to rewrite the previous theorem in terms of functions
on the hypercube $\zoo^{2n}$ rather than the set
$\zoon\times\{0,1,2,\dots,n\}.$

\begin{theorem}
\label{thm:univariate-reduction}
Put $k=\lfloor\alpha n\rfloor,$ where $\alpha>0$ is the absolute
constant from Theorem~\textup{\ref{thm:mu_s-constructed}}.
Choose $w_1,w_2,\dots,w_n$ uniformly at random from 
$\{0,1,\dots,2^{k+1}-1\}.$ Define
$f\colon\zoo^{2n}\to\moo$ by 
\begin{align*}
f(x) = \sign \left( \frac12 + \sum_{i=1}^n w_ix_i -
2^{k+1}\sum_{i=n+1}^{2n} x_i\right).
\end{align*}
Then with probability at least $1-\e^{-n/3}$ over the choice of
$w_1,w_2,\dots,w_n,$ one has 
\begin{align*}
R^+(f,d)\geq R^+(\{\esses\},d), && d=0,1,\dots,k.
\end{align*}
\end{theorem}

\begin{proof}
Immediate from Proposition~\ref{prop:symm-rational}
and Theorem~\ref{thm:univariate-reduction-symm}.
\end{proof}

\section{Main Result and Generalizations}
\label{sec:main}

We now combine the newly obtained result on rational
approximation with known results from Section~\ref{sec:prelim} to
prove the main theorem of this work.

\begin{theorem}[\rm Main result]
\label{thm:main-detailed}
Fix sufficiently small absolute constants $\alpha>0$ and
$\beta=\beta(\alpha)>0.$
Choose integers $w_1,w_2,\dots,w_n\in
\{0,1,\dots,2^{\lfloor\alpha n\rfloor+1}-1\}$ 
uniformly at random.
Then with probability at least $1-\e^{-n/3},$ 
the function $f\colon \zoo^{2n}\to\moo$ given by
\begin{align*}
f(x) = \sign\left(\frac12 + \sum_{i=1}^n w_ix_i -
2^{\lfloor\alpha n\rfloor +1}\sum_{i=n+1}^{2n} x_i\right)
\end{align*}
obeys
\begin{align}
\degthr(f\wedge f) \geq \lfloor \beta n\rfloor.
\label{eqn:main-beta}
\end{align}
\end{theorem}

\begin{proof}
Theorem~\ref{thm:univariate-reduction} shows that 
with probability at least $1-\e^{-n/3}$ over the choice of
$w_1,w_2,\dots,w_n,$ one has
\begin{align}
R^+(f,d)\geq
R^+(S,d),
    && d=0,1,\dots,\lfloor\alpha n\rfloor,
	\label{eqn:univariate-reduction}
\end{align}
where $S=\{\pm1,\pm2,\pm3,\dots,\pm2^{\lfloor\alpha n\rfloor}\}$
and $\alpha>0$ is the absolute constant from
Theorem~\ref{thm:mu_s-constructed}.  In the remainder of the
proof, we will condition on this event.

Suppose now that $\degthr(f\wedge f)<\lfloor \beta n\rfloor,$
where $\beta$ is a constant to be chosen later subject to
$0<\beta<\alpha/4.$ Then Theorem~\ref{thm:sherstov-degthr-R}
implies that $R^+(f,\lfloor4\beta n\rfloor)<1/2,$ which in view of
(\ref{eqn:univariate-reduction}) leads to $R^+(S,\lfloor4\beta
n\rfloor) <1/2.$ The last inequality violates
Theorem~\ref{thm:R+sign} for small enough $\beta>0.$ Thus,
(\ref{eqn:main-beta}) holds for $\beta$ small enough.
\end{proof}

Recall that the technical crux of this paper is an
optimal lower bound for the rational approximation of a
halfspace.  We will have occasion to appeal to this
result again, and for this reason we formulate it as a
theorem in its own right.

\begin{theorem}
\label{thm:R-hn}
A family of halfspaces $h_n\colon\zoon\to\moo,$
$n=1,2,3,\dots,$
exists such that
\begin{align}
R^+(h_n,d) = 1-\exp\BRACES{-\Theta\PARENS{\frac nd}}
, \qquad d=1,2,\dots,\Theta(n).
\label{eqn:R-hn}
\end{align}
\end{theorem}

\begin{proof}
The lower bound in (\ref{eqn:R-hn})
is immediate from Theorem~\ref{thm:univariate-reduction} 
and the univariate lower bounds in Theorem~\ref{thm:R+sign}.

Next, every halfspace $h_n\colon\zoon\to\moo$
constructed in Theorem~\ref{thm:univariate-reduction} trivially obeys
$R^+(h_n,1)<1-\exp\{-\Theta(n)\}.$ For $0<\xi<1,$ Newman's
classical work~\cite{newman64rational} shows that
$R^+([-1,-\xi]\cup[\xi,1],d)\leq 1-\xi^{\Theta(1/d)},$
whence by composition of the approximants one obtains the upper
bound in (\ref{eqn:R-hn}).
\end{proof}

\subsection*{Mixed intersection.} 
Theorem~\ref{thm:main-detailed} shows that the
intersection of two halfspaces has the asymptotically 
highest threshold degree. At the same time,
Beigel et al.~\cite{beigel91rational} showed that the
intersection of a constant number of majority functions
on $\zoon,$ which are particularly simple halfspaces,
has threshold degree $O(\log n).$ We now derive a lower
bound of $\Omega(\sqrt{n \log n})$ on the threshold
degree of the intersection of a halfspace and a
majority function, which improves quadratically on the
previous bound in~\cite{sherstov09hshs} and essentially
matches the upper bound, $O(\sqrt n\log n),$ given
below in Remark~\ref{rem:h-majn}.

\begin{theorem}
\label{thm:hmajn}
A family of halfspaces $h_n\colon\zoon\to\moo,$
$n=1,2,3,\dots,$
exists such that
\begin{align}
\degthr(h_n\wedge \MAJ_n) = \Theta(\sqrt {n\log n}).
\label{eqn:degthr-h-maj}
\end{align}
\end{theorem}

\begin{proof}
The lower bound in (\ref{eqn:degthr-h-maj}) is immediate from
Theorems~\ref{thm:approx-maj}, \ref{thm:sherstov-degthr-R}, 
and~\ref{thm:R-hn}.  The upper bound in (\ref{eqn:degthr-h-maj})
is immediate from Theorems~\ref{thm:approx-maj},
\ref{thm:rational-is-possible}, and~\ref{thm:R-hn}.
\end{proof}

\begin{remark}
\label{rem:h-majn}
The construction of Theorem~\ref{thm:hmajn} is essentially best
possible in that every sequence of halfspaces
$h_n\colon\zoon\to\moo,$
$n=1,2,3,\dots,$ obeys 
\begin{align}
\degthr(h_n\wedge\MAJ_n)=O(\sqrt n\log n).
\label{eqn:general-hmajn-upper}
\end{align}
To derive this upper bound, recall that
$R^+(h_n,1)<1-\exp\{-\Theta(n\log n)\}$ for every halfspace
$h_n\colon\zoon\to\moo,$ by a classical result due to
Muroga~\cite{muroga71threshold}. Since
$R^+([-1,-\xi]\cup[\xi,1],d)<1-\xi^{\Theta(1/d)}$ for
$0<\xi<1$ by Newman~\cite{newman64rational}, we obtain by composition
of approximants that $R^+(h_n,d)<1-\exp\{-\Theta(\{n\log n\}/d)\}.$
This settles (\ref{eqn:general-hmajn-upper}) in view of
Theorems~\ref{thm:approx-maj}
and~\ref{thm:rational-is-possible}.
\end{remark}

\subsection*{Threshold density.}
In addition to threshold degree, several other
complexity measures are of interest when
sign-representing Boolean functions by real
polynomials.  One such complexity measure is
\emph{density,} i.e., the least $k$ for which a given
function can be sign-represented by a linear
combination of $k$ parity functions.  Formally, for a
given function $f\colon\zoon\to\moo,$ the
\emph{threshold density} $\dns(f)$ is the minimum size
$\abs{\Scal}$ of a family
$\Scal\subseteq\Pcal(\oneton)$ such that
\begin{align*}
f(x) \equiv \sign\PARENS{\sum_{S\in\Scal}^{\phantom{S\in\Scal}}
\lambda_S \chi_S(x)}
\end{align*}
for some reals $\lambda_S,$ $S\in\Scal.$ It is clear
from the definition that $\dns(f)\leq 2^n$ for all
functions $f\colon\zoon\to\moo,$ and we will show that
the intersection of two halfspaces on $\zoon$ has
threshold density $2^{\Theta(n)}.$ 

\newcommand{\op}{\text{\rm KP}}
To this end, we recall an elegant technique for converting
Boolean functions with high threshold degree into Boolean
functions with high threshold density, due to Krause and
Pudl{\'a}k~\cite[Prop.~2.1]{krause94depth2mod}.  Their
construction sends a function $f\colon\zoon\to\moo$ to the
function $f^\op\colon(\zoon)^3\to\moo$ given by 
\begin{align*}
f^\op(x,y,z) =
f(\dots, (\overline{z_i}\wedge x_i)\vee(z_i\wedge y_i), \dots).
\end{align*}

\begin{theorem}[\rm Krause and Pudl{\'a}k]
\label{thm:degree-length}
Every function $f\colon\zoon\to\moo$ obeys
\begin{align*}
 \dns(f^\op) \geq 2^{\degthr(f)}.
\end{align*}
\end{theorem}

\noindent
We are now in a position to obtain the desired density results.

\begin{theorem}
\label{thm:dns}
A family of halfspaces $h_n\colon\zoon\to\moo,$
$n=1,2,3,\dots,$
exists such that
\begin{align}
\dns(h_n\wedge h_n) &\geq \exp\{\Theta(n)\},
\label{eqn:hh}\\
\dns(h_n\wedge \MAJ_n) &\geq \exp\{\Theta(\sqrt {n\log n})\}.
\label{eqn:hmajn}
\end{align}
\end{theorem}

\begin{proof}
The parity of several parity functions is another parity
function. As a result,
\begin{align}
\max_{h_n}\{\dns(h_n\wedge h_n)\} \geq \max_F \{\dns(F\wedge F)\},
\label{eqn:h-F}
\end{align}
where the maximum on the left is over all halfspaces
$h_n\colon\zoon\to\moo$ and the maximum on the right is over
arbitrary functions $F\colon\zoo^m\to\moo$ (for arbitrary $m$)
such that $\dns(F)\leq n.$ For each $n=1,2,3,\dots,$
Theorem~\ref{thm:main-detailed} ensures the existence of a
halfspace $f_n\colon\zoon\to\moo$ with $\degthr(f_n\wedge
f_n)\geq\Omega(n).$ By Theorem~\ref{thm:degree-length}, the
function $(f_n\wedge f_n)^\op={f_n}^\op\wedge {f_n}^\op$ has
threshold density $\exp\{\Omega(n)\}.$ Since $\dns({f_n}^\op)\leq
4n+1,$ the right member of (\ref{eqn:h-F}) is at least
$\exp\{\Omega(n)\}.$ 

This completes the proof of (\ref{eqn:hh}).
The proof of (\ref{eqn:hmajn}) is closely analogous, with
Theorem~\ref{thm:hmajn} used instead of
Theorem~\ref{thm:main-detailed}.
\end{proof}

%\begin{remark}
The lower bounds in Theorem~\ref{thm:dns} are essentially
optimal. Specifically, (\ref{eqn:hh}) is tight for
trivial reasons, whereas the lower bound (\ref{eqn:hmajn}) nearly
matches the upper bound of $\exp\{\Theta(\sqrt n\log^2n)\}$ that
follows from (\ref{eqn:general-hmajn-upper}).

We also note that 
Theorem~\ref{thm:degree-length}
readily generalizes to linear combinations of
conjunctions rather than parity functions. In
other words, if a function $f\colon\zoon\to\moo$ has threshold
degree $d$ and $f^\op(x,y,z)\equiv \sign(\sum_{i=1}^N\lambda_i
T_i(x,y,z))$ for some conjunctions $T_1,\dots,T_N$ of
the literals $x_1,y_1,z_1,\dots,x_n,y_n,z_n,$ $\neg x_1,\neg
y_1,\neg z_1,\dots,\neg x_n,\neg y_n,\neg z_n,$ then $N\geq
2^{\Omega(d)}.$ With this remark in mind,
Theorem~\ref{thm:dns} and its proof readily carry over to this
alternate definition of density.
%\end{remark}

%\begin{remark}[On explicitness]
%Given the learning-theoretic motivation for this paper, our
%interest was limited to the existential statement in
%Theorem~\ref{thm:main-detailed}, and
%we have not pursued an explicit construction of two
%halfspaces on $\zoon$ with high threshold degree. In
%particular, the best explicit result remains the
%$\Omega(\sqrt n)$ lower bound on the threshold degree
%obtained previously in~\cite{sherstov09hshs}.
%\end{remark}

\section*{Acknowledgments}

The author is thankful to Adam Klivans, Ryan O'Donnell,
Rocco Servedio, and the anonymous reviewers for their
feedback on this manuscript.

{
\renewcommand{\baselinestretch}{0.95}
    \small
\bibliographystyle{abbrv}%alphanumm}
\bibliography{%
/Users/sasha/bib/general,%
/Users/sasha/bib/fourier,%
/Users/sasha/bib/cc,%
/Users/sasha/bib/learn,%
/Users/sasha/bib/my}
}

\end{document}